\documentclass[12pt]{article}
\usepackage{amsmath}
\usepackage{amssymb}
\usepackage{amsthm}
\usepackage{tikz}
\usetikzlibrary{decorations.markings,decorations.pathreplacing,calc}
\usetikzlibrary{arrows.meta}
\usepackage{algorithm}
\usepackage{algpseudocode}
\usepackage{thm-restate}
\usepackage{verbatim}
\usepackage[framemethod=tikz]{mdframed}
\usepackage{complexity}

\usepackage[margin=1in]{geometry}
\usepackage{subfig}
\usepackage[utf8]{inputenc}

\usepackage[colorlinks = true]{hyperref}

\usepackage{xcolor}
\definecolor{darkred}  {rgb}{0.5,0,0}
\definecolor{darkblue} {rgb}{0,0,0.5}
\definecolor{darkgreen}{rgb}{0,0.5,0}

\hypersetup{
  urlcolor   = blue,         
  linkcolor  = darkblue,     
  citecolor  = darkgreen,    
  filecolor  = darkred       
}

\newcommand{\be}{\begin{equation}}
\newcommand{\ee}{\end{equation}}
\newcommand{\ba}{\begin{array}}
\newcommand{\ea}{\end{array}}
\newcommand{\bea}{\begin{eqnarray}}
\newcommand{\eea}{\end{eqnarray}}

\newcommand{\calL}{{\cal L }}

\newcommand{\calK}{{\cal K }}

\newcommand{\calX}{{\cal X }}
\newcommand{\calY}{{\cal Y }}
\newcommand{\FF}{\mathbb{F}}
\newcommand{\ZZ}{\mathbb{Z}}

\newcommand{\im}[1]{{\mathrm{Im}{(#1)}}}

\newcommand{\bbox}[1]{\mathrm{Box}({#1})}

\newtheorem{dfn}{Definition}

\newtheorem{claim}{Claim}
\newtheorem{lemma}{Lemma}

\newtheorem{theorem}{Theorem}
\newtheorem*{theorem*}{Theorem}

\newtheorem*{hlf}{Hidden Linear Function Problem}
\newtheorem*{2dhlf}{2D Hidden Linear Function Problem}
\renewcommand\footnotemark{}

\begin{document}
\title{Quantum advantage with shallow circuits}
\author{Sergey Bravyi$^{1}$, David Gosset$^{1}$, Robert K{\"o}nig$^{2}\thanks{sbravyi@us.ibm.com, dngosset@us.ibm.com, robert.koenig@tum.de}$\\ \textit{$^{1}$IBM T.J. Watson Research Center}\\\textit{$^{2}$ Institute for Advanced Study \& Zentrum Mathematik,}\\\textit{ Technische Universit\"{a}t M\"unchen}}
\date{}
\maketitle

\abstract{We prove that constant-depth quantum circuits are more powerful than their classical counterparts. 
To this end we introduce a non-oracular version of the Bernstein-Vazirani problem which we call the 2D Hidden Linear Function problem. An instance of the problem is specified by a quadratic 
form $q$ that maps $n$-bit strings to integers modulo four. The goal is to identify a linear boolean function which describes the action of $q$ on a certain subset of $n$-bit strings. 
We prove that any classical probabilistic circuit composed of bounded fan-in gates that solves
the 2D Hidden Linear Function problem with high probability must have depth logarithmic in $n$. In contrast, we show that this problem can be solved with certainty by a constant-depth quantum circuit 
composed of one- and two-qubit gates acting locally on a two-dimensional grid.
}


\section{Introduction}

Parallel  algorithms lay the foundation for modern  high-performance computing.
Many problems of practical importance such as  
Monte Carlo simulation, 
computing the  rank, the determinant,  and the inverse of a matrix
lend themselves naturally to parallelization~\cite{csanky1976fast,borodin1982fast}.
Of particular interest to us are the problems that 
can be solved on a parallel machine in a  constant time independent of the problem size
using a polynomial number of  processors. 
The class of such problems, known as $\NC^0$,
captures the computational power of constant-depth circuits 
with bounded fan-in gates~\cite{arora2009computational}.

The success of classical parallel algorithms motivates the study of their quantum counterparts. 
Quantum parallel algorithms running in a constant time take as input a classical
bit string, apply a constant-depth quantum circuit composed of one- and two-qubit gates,
and output a random bit string   obtained by measuring each qubit in the $0,1$-basis.  
For brevity, we shall refer to such computations as Shallow Quantum Circuits (SQC). 

Although SQCs are a highly restricted form of quantum computation, there is hope that they may outperform classical computers in certain tasks. A pioneering work by Terhal and DiVincenzo~\cite{terhal2002adaptive} gave the first evidence in this direction by showing that SQCs may be hard to simulate classically.  Quite recently, Bermejo-Vega et al.~\cite{bermejo2017architectures}, building on earlier studies of so-called IQP circuits \cite{bremner2016average,bremner2016achieving,farhi2016quantum}, gave further evidence that the output distribution of SQCs may be hard to sample classically even if a constant statistical error is allowed;
see also Ref.~\cite{gao2016quantum}.  A more powerful model of computation consisting of logarithmic-depth quantum circuits assisted by polynomial-time classical computation is known to be capable of solving hard problems such as factoring~\cite{cleve2000fast}.

Parallelism and circuit depth are important considerations when designing quantum algorithms that can be executed in the near-future on small quantum computers that may lack error correction capabilities~\cite{temme2016error,li2016efficient,boixo2016characterizing}. While the overhead for encoding and manipulating quantum data fault-tolerantly is asymptotically small~\cite{gottesman2013overhead,bombin2015gauge}, it remains
prohibitive for current technology. A quantum computation without error correction can only compute for a constant
amount of time before the qubits decohere and the entropy builds up~\cite{pastawski2009long}.
In this situation one may wish to parallelize the computation as much as possible to fit within the coherence time.

What can we hope to prove concerning the computational power of SQCs~?
A rigorous proof that SQCs outperform polynomial-time classical algorithms for some computational task
is arguably beyond our reach (as it would imply a separation between the complexity classes
$\BQP$ and $\BPP$). In this paper we  set a less ambitious goal and pose the following question:
\newline
\newline
\textit{Can constant-depth quantum circuits solve a computational problem that constant-depth classical circuits cannot?}
\newline

Put differently, we ask whether constant-time parallel quantum algorithms are more powerful than their
classical probabilistic counterparts. 
We show  that the answer to the above question is YES, even if the
quantum circuit is composed of nearest-neighbor gates acting on a 2D grid
whereas the only restriction on the constant-depth classical (probabilistic) circuit is 
having a bounded fan-in. In particular, the gates in the classical circuit may be long-range (i.e., they need not be geometrically local in 2D or otherwise) and may have unbounded fan-out. 
We emphasize that our result constitutes a provable separation and does not rely on any conjectures or assumptions concerning complexity classes. Formally, our result implies that there is a search (relational) problem solved by SQCs but not by $\NC^{0}$ circuits, even if we allow the classical circuit access to random input bits drawn from an arbitrary distribution depending on the input size.

The computational problem that we use to establish the above separation
is a simple linear algebra task concerning quadratic forms over the binary field,
see Section~\ref{sec:hlf} for formal definitions. We call it the 2D Hidden Linear Function problem.  The input is a binary vector $b\in \{0,1\}^n$ and 
 a binary $n\times n$ matrix $A$. Here the rows and columns of $A$ are labeled by sites
of a 2D grid with $n$ sites, and $A$ is sparse in the sense that $A_{i,j}=0$ unless the sites
$i,j$ are nearest neighbors on the grid. The pair $A,b$ defines a 
quadratic form $q\, : \, \FF_2^n \to \ZZ_4$  such that
\[
q(x)=2\sum_{1\le \alpha< \beta\le n}  A_{\alpha,\beta}\, x_\alpha x_\beta + \sum_{\alpha=1}^n b_\alpha x_\alpha,
\]
where $x_1,\ldots,x_n\in \FF_2$ are binary variables. 
Define the set 
\begin{equation}
\calL_q=\{ x\in \FF_2^n \, : \,  q(x\oplus y)=q(x) + q(y) \quad \forall y\in \FF_2^n\}. 
\label{eq:sumq}
\end{equation}
Here $x\oplus y$ denotes addition of binary vectors modulo two, while $q(x)+q(y)$ in Eq.~\eqref{eq:sumq} is evaluated modulo four. We show that  the restriction of $q(x)$ onto $\calL_q$
is always a linear form, that is, there exists a vector $z\in \FF_2^n$
such that
$q(x)=2\sum_{\alpha=1}^n z_\alpha x_\alpha$ for all $x\in \calL_q$.
The problem is to find such a vector $z$ (which may be  non-unique).
This problem can be viewed as an non-oracular
version of the well-known Bernstein-Vazirani problem~\cite{BV}, where the goal is to learn 
a hidden linear function specified by an oracle. In our case there is no oracle and the linear function is hidden inside the quadratic form $q$ which is explicitly
specified by its coefficients $A,b$. 

Our results can be summarized as follows.
\begin{theorem*}
Any classical probabilistic circuit with bounded fan-in gates
that solves the 2D Hidden Linear Function problem with success probability greater than $7/8$ 
must have depth $\Omega(\log{n})$. In contrast the problem can be solved with certainty
by a constant-depth quantum circuit composed of one- and two-qubit gates acting locally on a 2D grid.
\end{theorem*}

Our main technical contribution is the depth lower bound for classical circuits which solve the 2D Hidden Linear Function problem, see Section~\ref{sec:lowerbound}.
The proof exploits quantum nonlocality -- a form of correlation present in 
the measurement statistics of  entangled quantum states that cannot be reproduced by 
local hidden variable models~\cite{mermin,ghz} even if a limited amount of communication 
is allowed~\cite{barrett2007}. By leveraging the result of Ref.~\cite{barrett2007}
we prove that  a class of entangled quantum states known as cluster states (or 2D graph states) ~\cite{raussendorf2001one,briegel2001persistent}
exhibits a particularly strong form of  nonlocality that  cannot be reproduced by constant-depth classical circuits. At the same time, we show that such nonlocality is present in the input-output correlations of the 2D Hidden Linear Function problem.
We also prove that this problem  can be 
solved by a simple SQC that consists of two layers of single-qubit gates, and
a unitary operator $U_q$ such that $U_q|x\rangle =i^{q(x)}|x\rangle$. The latter can be easily implemented by a constant-depth circuit which takes as input the coefficients $A,b$ which specify $q$, see Section~\ref{sec:hlf}  for details.

Quantum algorithms for learning and testing  properties of   Boolean functions have a
long history~\cite{BV,deutsch1992rapid,atici2007quantum,aaronson2010bqp,gavinsky2011quantum,cross2015quantum}.
Most of these algorithms (including ours) work by sampling a probability distribution related to the 
Fourier transform of the function of interest. For example, 
an analogue of the Bernstein-Vazirani problem for quadratic forms over the
binary field has been previously considered by R{\"o}tteler~\cite{rotteler2009quantum}. That work describes a quantum algorithm that learns a quadratic form with $n$ binary variables
specified by an oracle using only $O(n)$ queries to the oracle. We note however that 
the  problem considered in Ref.~\cite{rotteler2009quantum} is not directly related to the 
one studied in the present paper. Our problem is more closely related to the Bernstein Vazirani problem \cite{BV} or,  more generally,  the Fourier Fishing problem introduced by Aaronson~\cite{aaronson2010bqp,Aaronson16}. Fourier Fishing is a search problem where the goal is to identify a bit sting $z$ such that the Fourier transform of a given  Boolean function has non-negligible weight on $z$. We shall see that a string $z$ is a solution of the 2D Hidden Linear Function problem for a quadratic form $q$ iff the Fourier transform of $q$ has a non-zero
weight on $z$, see Lemma~\ref{lem:pz} in Section~\ref{sec:hlf}.

To the best of our knowledge, quantum analogues of low-depth circuit classes $\NC$ and $\NC^k$
were first introduced by Moore and Nilsson~\cite{moore2001parallel}.
This work also provided techniques for parallelizing quantum circuits.
Hoyer and Spalek~\cite{hoyer2005quantum} showed that combining SQCs with the
so-called quantum fan-out gates (CNOTs with a single control and multiple target qubits) 
gives a class of quantum circuits that is powerful enough to solve the factoring
problem, if assisted by polynomial-time classical computation.
An alternative characterization of this class of circuits in terms of 
measurement-based computations was later given by Browne et al.~\cite{browne2010computational}.
Finally, Terhal and DiVincenzo~\cite{terhal2002adaptive} showed
that any quantum circuit composed of one- and two-qubit gates
can be compressed to a constant-depth  if the ability to post-select measurement outcomes is given. 
This implies that post-selective SQCs can solve any problem 
in the complexity class $\PostBQP=\PP$, see Ref.~\cite{aaronson2005quantum}.
On the negative side, Eldar and Harrow have
recently shown~\cite{eldar2015local}
 that SQCs are not capable of preparing 
certain entangled  states associated with good quantum codes
that  can be prepared by polynomial-size quantum circuits.
Furthermore, it is known that depth-$2$ SQCs can be efficiently simulated classically
in polynomial time~\cite{terhal2002adaptive}.
Superpolynomial-time classical algorithms for simulating more general SQCs are discussed in Ref.~\cite{Aaronson16}.

The remainder of the paper is organized as follows. We  begin in Section \ref{sec:prelim} by establishing some definitions and notation. In Section \ref{sec:hlf} we discuss computational problems where the goal is to find a hidden linear boolean function. In particular, we describe the 2D Hidden Linear Function problem and show that it is solved deterministically by SQCs. In Section \ref{sec:lowerbound} we prove that the 2D Hidden Linear Function problem cannot be solved by classical circuits of low depth. We conclude in Section \ref{sec:conc} with some open questions.


\section{Preliminaries \label{sec:prelim}}
Here we describe the classes of circuits considered in this paper and review the definition of quantum graph states.

\paragraph{Quantum circuits} A quantum circuit of depth $d$ consists of a sequence of $d$ layers of one- and two-qubit gates, such that  gates within a given layer act on disjoint sets of qubits. In other words, the unitary  implemented by the circuit is a product $U_d\ldots U_2U_1$ where each $U_j$ is a tensor product of one- and two-qubit gates which act nontrivially on disjoint sets of qubits. We shall use the Clifford+T gate 
set, that is, we assume that all single-qubit gates are either\footnote{Our definition of the $S$-gate differs from the standard one (which has $+i$ instead of $-i$).}
\[
H=\frac{1}{\sqrt{2}}\left( \begin{array}{cc}1 & 1\\1 & -1\end{array}\right), \qquad S=\left( \begin{array}{cc}1 & 0\\0 & -i\end{array}\right), \quad \text{or} \quad  T=\left( \begin{array}{cc}1 & 0\\0 & e^{i\pi/4}\end{array}\right)
\]
and all two-qubit gates are controlled-$Z$ gates $CZ=|0\rangle\langle 0|\otimes I+|1\rangle\langle 1|\otimes Z$. Here and below $X,Y,Z$ denote the single-qubit Pauli matrices,
\[
X=\left( \begin{array}{cc}0 & 1\\1 & 0\end{array}\right), \qquad
Y=\left( \begin{array}{cc}0 & -i\\i & 0\end{array}\right), \qquad
Z=\left( \begin{array}{cc}1 & 0\\0 & -1\end{array}\right).
\]

\paragraph{Classical circuits}
A classical circuit is specified by a directed acyclic graph in which each vertex is either an input (if it has in-degree $0$), an output (if it has out-degree $0$), or a gate. For each gate we must also specify a boolean function $\{0,1\}^k \to \{0,1\}$
which is computed by the gate, where $k$ is the in-degree or ``fan-in" of the gate. The output bit computed by a gate is copied to all outgoing edges of this gate. 
 The out-degree or ``fan-out" of a gate is the number of times the output of the gate is used in the remainder of the circuit. 
The depth $d$ of a classical circuit is the maximum number of gates along a path from an input to an output. 

Let $\mathcal{C}$ be a classical circuit with input $x\in \{0,1\}^m$ and output $z=F(x)\in \{0,1\}^n$. We will consider probabilistic circuits in which a subset of input bits may be chosen at random, that is $x=x'r$ where $r\in \{0,1\}^\ell$ is a random string drawn from some (arbitrary) distribution $\rho(r)$ and $x'\in \{0,1\}^{m-\ell}$. We shall often identify a variable $x_j$ or $z_k$ with its index $j$ or $k$ respectively.
\begin{dfn}
A pair of variables $x_j,z_k$ is said to be correlated iff there exists $y\in \{0,1\}^m$ such that flipping the $j$th bit of $y$ flips the $k$th bit of $F(y)$.
\end{dfn}
Note that $x_j,z_k$ can be correlated only if the graph describing the circuit contains a path from $x_j$ to $z_k$.
\begin{dfn}[Lightcone]
Let $\mathcal{C}$ be a classical circuit. For each input bit $x_j$ define its light cone $L_\mathcal{C}(x_j)$ to be the set of output variables correlated with $x_j$. Likewise, for each output bit $z_k$ define its light cone $L_{\mathcal{C}}(z_k)$ to be the set of input variables correlated with $z_k$. 
\end{dfn}
In this paper we are interested in circuits with constant depth such that all gates have fan-in at most $K=O(1)$. This class of circuits is known as $\NC^{0}$. Any such circuit computes a local function in the sense that each output bit can only be correlated with a constant number of input bits. Indeed, if $\mathcal{C}$ has depth $d$ then for all $i$ we have
\begin{equation}
|L_{\mathcal{C}}(z_i)|\leq K^d.
\label{eq:lightconesize}
\end{equation}

\paragraph{Graph states} Let $G=(V,E)$ be a finite simple graph. Define an associated $|V|$-qubit graph state $|\Phi_{G}\rangle$ by
\begin{equation}
|\Phi_{G}\rangle=\left(\prod_{\{i,j\}\in E} CZ_{ij}\right)H^{\otimes |V|}|0^{|V|}\rangle.
\label{eq:graphstatedefn}
\end{equation}
Here and below the subscript of a gate indicates the qubit(s) acted upon by this gate, 
and we use a shorthand notation $|0^m\rangle \equiv |0\rangle^{\otimes m}$.
The graph state $|\Phi_{G}\rangle$ is a stabilizer state with stabilizer group generated by 
\begin{equation}
g_v=X_v \left(\prod_{w: \{w,v\}\in E} Z_w\right) \qquad v\in V.
\label{eq:stabgroup}
\end{equation}
In other words $|\Phi_{G}\rangle$ is the unique state satisfying $g_v |\Phi_G\rangle=|\Phi_G\rangle$ for all $v\in V$.


\section{Hidden linear function problems \label{sec:hlf}}

In this section we consider computational problems where the goal is to find a hidden linear function.  

Our starting point is the well-known Bernstein-Vazirani problem \cite{BV}. Here one is given oracle access to a linear 
boolean function $\ell: \mathbb{F}_2^n\rightarrow \mathbb{F}_2$ parameterized by a ``secret" bit string $z\in \{0,1\}^n$ 
\[
\ell(x)=z^Tx \mod 2  \qquad x\in\{0,1\}^n.
\]
Here and below $z^Tx\equiv \sum_{\alpha=1}^n z_\alpha x_\alpha  $ denotes the inner product of vectors. 
Bernstein and Vazirani showed that one can identify the linear function $\ell$ (i.e., find the secret bit string $z$) using just one quantum query to an oracle $U_\ell$ which performs the unitary
\[
U_\ell|x\rangle=(-1)^{\ell(x)}|x\rangle.
\]
Indeed, we have
\[
|z\rangle=H^{\otimes n}U_\ell H^{\otimes n} |0^n\rangle.
\]
On the other hand a classical algorithm with access to a classical oracle computing $\ell$ requires $n$ queries to obtain $z$.

In the Bernstein-Vazirani problem the oracle hides the linear function. The quantum speedup obtained is \textit{relative to the oracle} and is not guaranteed to translate into a real-world quantum advantage.  Where else (other than inside an oracle) can we hide a linear function? 

We will now see how to hide a linear boolean function inside a $\mathbb{Z}_4$-valued quadratic form.  In particular, we consider quadratic forms $q\, : \, \FF_2^n \to \ZZ_4$ such that
\begin{equation}
\label{q}
q(x)=2\sum_{1\le \alpha< \beta\le n}  A_{\alpha,\beta}\, x_\alpha x_\beta + \sum_{\alpha=1}^n b_\alpha x_\alpha,
\end{equation}
where $x_1,\ldots,x_n\in \{0,1\}$ are binary variables\footnote{Here for simplicity we only consider $b_{\alpha}\in \{0,1\}$. Our results could alternatively be proved using a more general definition of quadratic forms where we allow $b_{\alpha}\in \{0,1,2,3\}$. },
\begin{equation}
\label{Ab}
A_{\alpha,\beta}\in \{0,1\} \quad \mbox{and} \quad  b_\alpha \in  \{0,1\}.
\end{equation}
In the remainder of this Section all arithmetic is performed in  the ring $\ZZ_4$ (unless stated otherwise). We shall label elements of $\ZZ_4$ by $\{0,1,2,3\}$. Entrywise addition of vectors modulo $2$ will be denoted
$\oplus$.  Define the set 
\begin{equation}
\label{Lq}
\calL_q=\{ x\in \FF_2^n \, : \,  q(x\oplus y)=q(x) + q(y) \quad \forall y\in \FF_2^n\}.
\end{equation}
Now let us see how $q$ hides a linear boolean function.
\begin{lemma}
The set $\calL_q$ is a linear subspace of $\FF_2^n$
and $q(x)\in \{0,2\}$ for all $x\in \calL_q$.
 The restriction of $q$ to $\calL_q$
is a linear function, that is, there exists a vector $z\in  \{0,1\}^n$ such that  $q(x)=2z^T x$
for all $x\in \calL_q$.
\end{lemma}
\begin{proof}
Suppose $x,x'\in \calL_q$. Choose any $y\in \FF_2^n$. Then 
\[
q(x\oplus x'\oplus y)=q(x)+q(x'\oplus y) = q(x) + q(x') + q(y) = q(x\oplus x') + q(y).
\]
This proves $x\oplus x'\in \calL_q$, that is, $\calL_q$ is a linear subspace.
Note that $0=q(0)=q(x\oplus x)=2q(x)$ for any $x\in \calL_q$,
that is, $q(x)\in \{0,2\}$.
Define a function $l\, : \, \calL_q \to \FF_2$ by
\[
l(x)=\left\{ \ba{rcl}
1 & \mbox{if} & q(x)=2,\\
0 & \mbox{if} & q(x)=0.\\
\ea
\right.
\]
From  Eq.~(\ref{Lq}) one infers that $l(x)$ is linear modulo two, 
\[
l(x\oplus y) = l(x)\oplus l(y) \qquad \mbox{for all $x,y\in \calL_q$}.
\]
It follows that $l(x)=z^T x {\pmod 2}$ for some vector $z\in \{0,1\}^n$.
Thus $q(x)=2z^T x$ for all $x\in \calL_q$.
\end{proof}
The linear action of $q$ on the subspace $\mathcal{L}_q$ can be parameterized by a secret bit string $z\in\{0,1\}^n$. In contrast with the Bernstein-Vazirani problem, here $z$ is not unique because the hidden linear function is only defined on a subspace of $\mathbb{F}_2^{n}$. To see this, let $\mathcal{L}_q^{\perp}$ be the orthogonal complement of $\mathcal{L}_q$. Then for any valid secret string $z$, and any $y\in \mathcal{L}_q^{\perp}$, $z\oplus y$ is also a valid secret string (since $2(z\oplus y)^T x=2z^Tx+2y^T x=2z^Tx$ for $x\in \mathcal{L}_q$). We define a search problem where the goal is to find a valid secret string.

\begin{mdframed}
\begin{hlf}
The input is a quadratic form $q\, : \, \FF_2^n \to \ZZ_4$ specified by 
a matrix $A$ and a vector $b$ as in Eqs.~(\ref{q},\ref{Ab}).
A solution is a binary vector $z\in \{0,1\}^n$ such that 
$q(x)=2z^T x$ for all $x\in \calL_q$.
\end{hlf}
\end{mdframed}

A quantum circuit which solves the Hidden Linear Function problem is shown in Fig.~\ref{fig:hlf}. Here there are two input registers holding $b$ and $A$ as well as an $n$-qubit data register. The circuit involves controlled gates which implement 
\begin{equation}
CZ(A)=\prod_{1\leq i<j\leq n} CZ_{ij}^{A_{ij}} \qquad S(b)=\bigotimes_{j=1}^{n} S_j^{b_j}
\end{equation}
on the data register conditioned on the first two registers holding bit strings $b,A$ respectively. The $n$-qubit unitary $U_q=S(b)CZ(A)$ satisfies
\begin{equation}
U_q |x\rangle = i^{q(x)} |x\rangle,  \qquad
x\in \{0,1\}^n.
\label{eq:Uq}
\end{equation}
The output $z\in \{0,1\}^n$ of the circuit is therefore drawn from the distribution
\begin{equation}
\label{eq:prob}
p(z)=|\langle z| H^{\otimes n} U_q H^{\otimes n} |0^n\rangle|^2.
\end{equation}
The following Lemma asserts that the circuit from Fig.~\ref{fig:hlf} solves the Hidden Linear Function problem deterministically (i.e., its output $z$ is a solution with probability $1$).
\begin{lemma}
$p(z)>0$ iff $z$ is a solution of the Hidden Linear Function problem, that is,
$q(x)=2z^T x$ for all $x\in \calL_q$. Furthermore, $p(z)$ is the uniform distribution on the
set of all solutions $z$.
\label{lem:pz}
\end{lemma}
\begin{proof}
For any linear subspace $\calL\subseteq \FF_2^n$ and a vector $z\in \FF_2^n$ define  
a partial Fourier transform
\[
\Gamma(\calL,z)\equiv \sum_{x\in \calL} (-1)^{z^T x} \cdot i^{q(x)}.
\]
Then 
\begin{equation}
\label{GG}
p(z)=\frac1{4^n} \left| \Gamma(\FF_2^n,z)\right|^2.
\end{equation}
Choose any linear subspace $\calK\subseteq \FF_2^n$ such that 
\begin{equation}
\label{LK}
\FF_2^n = \calL_q  + \calK  \quad \mbox{and} \quad \calL_q \cap \calK=0.
\end{equation}
From Eq.~(\ref{Lq}) one infers that 
\begin{equation}
\label{GG1}
\Gamma(\FF_2^n, z)=\Gamma(\calL_q, z) \cdot \Gamma(\calK,z).
\end{equation}
\begin{claim}
$\Gamma(\calL_q,z)=|\calL_q|$ if   $z$ is a solution of the Hidden Linear Function Problem
and $\Gamma(\calL_q,z)=0$ otherwise. The number of solutions to the Hidden Linear Function Problem is $|\mathcal{L}_q^{\perp}|$.
\label{claim:partialfourier1}
\end{claim}
\begin{proof}
By Lemma~1 there exists a vector $y\in \FF_2^n$ such that $q(x)=2y^T x$ for all
$x\in \calL_q$. Then $i^{q(x)}=(-1)^{y^T x}$ and thus 
\[
\Gamma(\calL_q,z)=\sum_{x\in \calL_q} (-1)^{x^T(y\oplus z)}=\left\{
\ba{rcl}
|\calL_q| &\mbox{if} & y\oplus z\in \calL_q^\perp,\\
0 && \mbox{otherwise}. \\
\ea\right.
\]
Note that the first case,  $y\oplus z\in \calL_q^\perp$, occurs iff $z$ is a solution of 
the Hidden Linear Function Problem, since $y$ and $z$ have the same binary inner product with any vector
from $\calL_q$ iff $y\oplus z\in \calL_q^\perp$. Therefore the number of solutions is $|\mathcal{L}_q^{\perp}|$.
\end{proof}
\begin{claim}
$|\Gamma(\calK,z)|^2=2^{n}\cdot |\calL_q|^{-1}$ for all $z\in \FF_2^n$.  
\label{claim:partialfourier2}
\end{claim}
We provide a proof of Claim \ref{claim:partialfourier2} in Appendix \ref{app:correctness}. The statement of the lemma follows directly from Eqs.~(\ref{GG},\ref{GG1}) and Claims~\ref{claim:partialfourier1},\ref{claim:partialfourier2}.
\end{proof}
A simple corollary of Lemma \ref{lem:pz} is that the Hidden Linear Function problem can be solved in polynomial-time using a classical computer. Indeed, for a given input $(A,b)$ the circuit from Fig.~\ref{fig:hlf} applies a sequence of Clifford gates to the data register followed by measurement in the computational basis. A solution can be efficiently computed classically by 
simulating this circuit using the Gottesman-Knill Theorem (see for example Ref.~\cite{aaronson2004improved}).

Since $p(z)$ is proportional to the absolute value squared of the
Fourier transform of $i^{q(x)}$, see Eq.~(\ref{eq:prob}), we see that the Hidden Linear Function problem is a variant of the Fourier Fishing problem defined in Section~6.2 of Ref.~\cite{Aaronson16}. One difference is that instead of Boolean functions we consider $\ZZ_4$-valued functions. A second important difference is that we consider an explicit computational problem (as opposed to an oracular one).

\begin{figure}
\centering
\begin{tikzpicture}[scale=1]
\draw (0,0.5)--(1,0.5);
\draw (1,0) rectangle (2,1);
\draw (2,0.5)--(2.7,0.5);
\draw (2.7,0) rectangle (4.3,1);
\draw (4.3,0.5)--(5,0.5);
\draw (5,0) rectangle (6,1);
\draw (6,0.5)--(7,0.5);
\draw (7,0) rectangle (8,1);
\draw (8,0.5)--(9,0.5);

\draw (9,0) rectangle (10,1);
\draw[thick] (9.08,0.4) arc (160:20:0.45);
\draw[thick, ->](9.5,0.5)--(9.8,0.8);

\node at (1.5,0.5){$H^{\otimes n}$};
\node at (7.5,0.5){$H^{\otimes n}$};
\node at (3.5,0.5){$CZ(A)$};
\node at (5.5,0.5){$S(b)$};

\node at (-0.3,0.5) {$|0^n\rangle$};
\node at (-0.3,1.5) {$|A\rangle$};
\node at (-0.3,2.5) {$|b\rangle$};

\node at (9.3,1.5) {$|A\rangle$};
\node at (9.3,2.5) {$|b\rangle$};

\draw (0,1.5)--(9,1.5);
\draw (0,2.5)--(9,2.5);

\draw (3.5,1.5)--(3.5,1);
\node[circle, fill=black, inner sep=2pt] at (3.5,1.5){};

\draw (5.5,2.5)--(5.5,1);
\node[circle, fill=black, inner sep=2pt] at (5.5,2.5){};

\draw [thick,dashed] (2.5,-0.5) rectangle (6.2,3);
\draw [decorate,decoration={brace,amplitude=10pt}]  (2.5,3.2)--(6.2,3.2);
\node at (4.4,3.8){$U_q$};
\end{tikzpicture}
\caption{A quantum circuit which solves the Hidden Linear Function problem. Here $CZ(A)=\prod_{1\leq i<j\leq n} CZ_{ij}^{A_{ij}}$ and $S(b)=\bigotimes_{j=1}^{n} S_j^{b_j}$.  \label{fig:hlf}}
\end{figure}
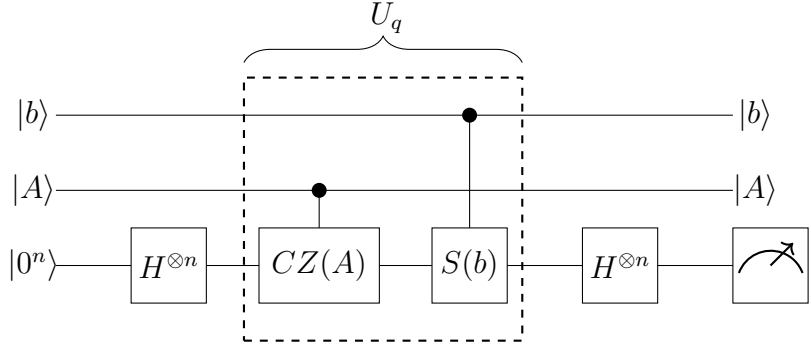

In general the circuit from Fig.~\ref{fig:hlf} may not be expressible as a constant-depth quantum circuit. We now restrict our attention to a subset of instances where $A$ has a certain 2D structure. We will see that for such instances the circuit from  Fig.~\ref{fig:hlf} can be decomposed as a constant-depth quantum circuit over the Clifford+T gate set.

In particular, let $G=(V,E)$ be the $N\times N$ grid graph. If we label vertices of $G$ by horizontal and vertical coordinates
\[
V=\left\{(i,j): 1\leq i,j\leq N\right\}
\]
then it has vertical edges $e=\{v,w\}\in E$ between all pairs $v=(i,j),w=(i,j+1)$ and horizontal edges between all pairs $v=(i,j), w=(i+1,j)$. Then $|V|=N^2$ and $|E|=2N(N-1)$. Let us consider quadratic forms Eq.~(\ref{q}) where the matrix $A$ specifies a subgraph of $G$. More precisely, we assume that $A$ is an $N^2\times N^2$ matrix whose rows and columns are labeled by vertices of $G$. We require that $A_{v,w}=0$ unless the pair $\{v,w\}$ is an edge of $G$.
We shall parameterize such a matrix by $|E|$ binary variables $A_e\in \{0,1\}$ such that $A_e=A_{v,w}=1$ iff $e=\{v,w\}\in E$. 
Now we can rewrite Eq.~(\ref{q}) as 
\begin{equation}
\label{q2D}
q(x)=2\sum_{e=\{v,w\}\in E}  A_e x_{v} x_w + \sum_{v\in V} b_v x_v, \qquad A\in \{0,1\}^{|E|} \quad b\in \{0,1\}^{|V|}.
\end{equation}
\begin{mdframed}
\begin{2dhlf}
Let $G=(V,E)$ be the $N\times N$ grid. The input is a $\ZZ_4$-valued quadratic form $q$ specified by vectors $A$ and $b$ as in Eq.~(\ref{q2D}). A solution is a binary vector $z\in \{0,1\}^{|V|}$ such that 
$q(x)=2z^T x$ for all $x\in \calL_q$.
\end{2dhlf}
\end{mdframed}
We shall say that the above describes an instance of size $N$. The number of input bits in an instance of size $N$ 
is then $|V|+|E|=3N^2 -2N$. 
The number of output bits is $n=|V|=N^2$.

The quantum algorithm from Fig.~\ref{fig:hlf} solves the 2D Hidden Linear Function problem and can be implemented in constant depth:
\begin{theorem}
For each $N\geq 2$ there exists a quantum circuit $\mathcal{Q}_N$ of depth $d=O(1)$ which deterministically solves size-$N$ instances of the 2D Hidden Linear Function problem.
\end{theorem}
\begin{proof}
It suffices to show that the controlled-$S(b)$ and controlled-$CZ(A)$ gates in Fig.~\ref{fig:hlf} can be expressed as constant-depth quantum circuits composed of one- and two-qubit gates over the Clifford+T gate set. 

The controlled-$S(b)$ gate can be implemented by applying a layer of two-qubit 
controlled-$S$ gates
$CS=|0\rangle\langle0|\otimes I+|1\rangle\langle 1|\otimes S$ acting on disjoint pairs of qubits. 
The gate $CS$ can be expressed exactly (with constant depth, of course) over the Clifford+T gate set
(see for example Ref.~\cite{amy2013meet}), which shows that the controlled-$S(b)$ gate from
Fig.~\ref{fig:hlf} can be implemented in constant depth. 

The controlled-$CZ(A)$ gate can be written as a product of three-qubit controlled-controlled-Z (CCZ) gates:
\begin{equation}
\label{eq:CCZ}
\prod_{e=\{v,w\}\in E} I\otimes \bigg(|0\rangle \langle 0|_e\otimes I +|1\rangle \langle 1|_e \otimes CZ_{v,w}\bigg)
\end{equation}
(here the three registers are as in Fig.\ref{fig:hlf}--the first one holds $b$, the second one holds $A$ and the third one is the data register). We can partition the edges of the 2D grid into four disjoint layers $E=E_1\cup E_2\cup E_3 \cup E_4$ such that no edges within a given layer share a vertex. Thus Eq.~\eqref{eq:CCZ} can be written as a depth-$4$ circuit composed of $CCZ$ gates.  Each $CCZ$ gate can then be decomposed exactly as a product of (a constant number of) one- and two-qubit gates from the Clifford+T gate set~\cite{amy2013meet}. 
\end{proof}

The circuit of Fig.~\ref{fig:hlf} can be embedded into a two-dimensional grid of size $N\times N$
such that  each vertex  of the grid contains $O(1)$ qubits 
and all  two-qubit gates are geometrically local. Indeed, the circuit contains $n=|V|$ target qubits (the bottom register initialized in the $|0^n\rangle$
state) and $|V|+|E|$ control qubits (the top and the center registers initialized in the $|b\rangle$ and $|A\rangle$ states). 
We shall place the target qubit labeled by a vertex $v$
and the control qubit that holds $b_v$ at the vertex $v$ of the grid.
Place the control qubit that holds $A_e$ with a horizontal (resp. vertical) edge $e$
at the left (resp. bottom) endpoint of the edge $e$. Now each vertex contains at most four qubits
and each two-qubit gate couples qubits  located at the same vertex or a pair of 
nearest-neighbor vertices. Thus the 2D Hidden Linear Function problem can be solved by 
a constant-depth quantum circuit with geometrically local gates\footnote{We note that a physical
implementation of the circuit from Fig.~\ref{fig:hlf} would require
only $n=N^2$ qubits and use only (classically controlled) Clifford gates
$H, S,CZ$ with  nearest-neighbor $CZ$ gates.}. We now show that even this highly restricted class of SQCs is more powerful than classical constant-depth circuits. 

\begin{theorem}
The following holds for all sufficiently large $N$. Let $\mathcal{C}_N$ be a classical probabilistic circuit with fan-in at most $K$ which solves all size-$N$ instances of the 2D Hidden Linear Function problem with probability greater than $7/8$. Then the depth of $\mathcal{C}_N$ is at least 
\[
\frac{1}{8}\frac{\log(N)}{\log(K)}.
\]
\label{thm:lowerbound}
\end{theorem}
Here the classical circuit $\mathcal{C}_N$ takes input $A,b$ along with a random string $r$ drawn from some (arbitrary) probability distribution and its output $z\in \{0,1\}^{N^2}$ must be a solution to the given instance with probability greater than $7/8$. We prove the Theorem in the next Section.


\section{Nonlocality thwarts shallow classical circuits \label{sec:lowerbound}}
Constant-depth classical circuits with gates of bounded fan-in exhibit a kind of locality expressed by Eq.~\eqref{eq:lightconesize}. On the other hand it is well known that the correlations present in the measurement statistics of entangled quantum states exhibit \textit{quantum nonlocality}--that is, they cannot be reproduced by completely local functions in which every output bit depends only on a single associated input bit as well as a shared random string. In this Section we provide a proof of Theorem \ref{thm:lowerbound} which exploits this tension. We will see how quantum nonlocality--even in states generated by constant-depth quantum circuits-- thwarts simulation by classical circuits which are (a) completely local, (b) geometrically local in one dimension, and finally (c) ``constant-depth local" in the sense of Eq.~\eqref{eq:lightconesize}. 

Let us begin with a famous illustration \cite{ghz,mermin} of quantum nonlocality. Define the $3$-qubit GHZ state
\[
|\mathrm{GHZ}\rangle=\frac{1}{\sqrt{2}}\left(|000\rangle+|111\rangle\right),
\]
which satisfies
\begin{equation}
P|\mathrm{GHZ}\rangle=|\mathrm{GHZ}\rangle \qquad P\in \{X_1X_2X_3,-X_1Y_2Y_3,-Y_1X_2Y_3,-Y_1Y_2X_3\}.
\label{eq:ghzstab}
\end{equation}
Let $b\in \{0,1\}^3$ and consider the measurement outcomes $m\in \{-1,1\}^3$ obtained by measuring each qubit $j$ of $|\mathrm{GHZ}\rangle$ in either the $X$ basis (if $b_j=0$) or the $Y$ basis (if $b_j=1$). Eq.~(\ref{eq:ghzstab}) implies that the quantum measurement statistics satisfy
\begin{equation}
i^{b_1+b_2+b_3}m_1m_2m_3=1 \qquad \text{whenever}\qquad b_1\oplus b_2\oplus b_3=0. 
\label{eq:ghz}
\end{equation}
In contrast it is not hard to show that Eq.~\eqref{eq:ghz} cannot be satisfied by a \textit{local hidden-variable model} in which the measurement outcomes $m$ are completely local functions of the measurement settings $b$ and also may depend on a shared random string $r$, that is, $m_j=m_j(b_j,r)$ for $j=1,2,3$.

Barrett et al. \cite{barrett2007} described an extension of the GHZ example which we review (with some small modifications) in Section \ref{sec:1dnonlocal}. It shows that the statistics of single-qubit measurements on the graph state corresponding to an even length cycle posess geometrically nonlocal correlations. In other words, certain measurement outcomes are correlated with measurement settings (i.e., choice of single-qubit measurement bases) that are far away with respect to the shortest path on the cycle. These measurement statistics cannot be simulated by low-depth classical circuits which are geometrically local in one dimension.

Finally, in Section \ref{sec:2dnonlocal} we turn our attention to the 2D Hidden Linear Function problem. Let us now argue that the setting is not so different from the above two examples.  Recall that the circuit from Fig.~\ref{fig:hlf} produces a uniformly random solution $z\in \{0,1\}^n$ to a given instance $A,b$. Just like in the GHZ example, here the output $z$ is obtained by measuring each qubit of an entangled quantum state in one of two possible bases.  Indeed, observe that the first two gates in the circuit prepare the graph state $|\Phi_{\mathcal{A}}\rangle$ for the graph $\mathcal{A}$ with adjacency matrix $A$, and the rest of the circuit measures each qubit $v$ in either the $X$ basis (if $b_v=0$) or the $Y$ basis 
(if $b_v=1$)\footnote{To see this note that $HZH=X$ and $(HS)^{\dagger} Z(HS)=Y$.}.

To prove Theorem \ref{thm:lowerbound} we establish that the correlations between the input $A,b$ and the output $z$ in the 2D Hidden Linear Function problem have a strong form of nonlocality that cannot be reproduced by constant-depth probabilistic classical circuits of bounded fan-in. We first suppose $\mathcal{C}_N$ is a classical circuit which solves size-$N$ instances of the 2D Hidden Linear Function problem with high probability. We restrict our attention to instances where $A$ specifies a subgraph of the grid which is an even length cycle. For each such instance we can infer (from the Barrett et al. example) a geometrically nonlocal feature of the correlations generated by $\mathcal{C}_N$. We use a probabilistic argument and Eq.~\eqref{eq:lightconesize} to show that at least one of these features is absent if $\mathcal{C}_N$ has depth $o(\log(N))$. Thus we obtain the desired lower bound on the depth of $\mathcal{C}_N$.

\subsection{Geometric nonlocality in a 1D graph state \label{sec:1dnonlocal}}
\begin{figure}
\centering
\begin{tikzpicture}
\draw (-2,0)--(-1.3,1.4);
\draw (-0.7,2.6)--(0,4);
\draw (2,0)--(1.3,1.4);
\draw (0,4)--(0.7,2.6);
\draw (-2,0)--(-0.5,0);
\draw (0.5,0)--(2,0);
\node at (0,0){\Large{...}};
\node[rotate={62}]  at (-1,2){\Large{...}};
\node[rotate={-62}]  at (1,2){\Large{...}};
\node[circle, fill=black, inner sep=2pt] at (0,4){};
\node[circle, fill=black, inner sep=2pt] at (-2,0){};
\node[circle, fill=black, inner sep=2pt] at (2,0){};
\node[circle, fill=black, inner sep=2pt] at (1.3,0){};
\node[circle, fill=black, inner sep=2pt] at (-1.3,0){};

\node[circle, fill=black, inner sep=2pt] at (-1.6,0.8){};
\node[circle, fill=black, inner sep=2pt] at (1.6,0.8){};

\node[circle, fill=black, inner sep=2pt] at (-0.4,3.2){};
\node[circle, fill=black, inner sep=2pt] at (0.4,3.2){};
\draw [decorate,decoration={brace,amplitude=10pt}, rotate=180] (-1.6,0.3)--(1.6,0.3);
\node at (0,-1){\large{$B$}};
\node at (-2,2.4){\large{$L$}};
\node at (1.9,2.4){\large{$R$}};
\node at (-2.2,-0.4){\large{$u$}};
\node at (2.1,-0.4){\large{$w$}};
\node at (0,4.4){\large{$v$}};
\draw [decorate,decoration={brace,amplitude=10pt}] (-2,0.4)--(-0.4,3.6);
\draw [decorate,decoration={brace,amplitude=10pt}] (0.4,3.6)--(2,0.4);
\end{tikzpicture}
\caption{We consider a cycle $\Gamma$ of even length $M$ and three vertices $u,v,w$ such that all pairwise distances are even. The sides $L,R,B$ may have unequal lengths\label{fig:triangle}.}
\end{figure}
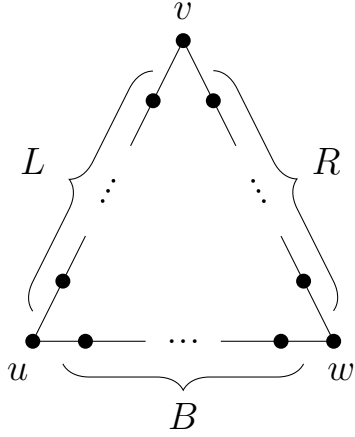
Here we present a variant of an example due to Barrett et al. \cite{barrett2007}. 
Let $\Gamma$ be the $M$-cycle graph with $M$ even. 
Suppose $u,v,w$ are vertices such that all pairwise distances between them are even. We shall say that a vertex $j$ is even (resp. odd) if it has even distance (resp. odd distance)
from $u,v,w$. It will be convenient to view $\Gamma$ as a triangle as shown in Fig.~\ref{fig:triangle}. We define sets $L,R,B$ of vertices for each of the three sides as shown. The vertices $u,v,w$ are not contained in any of these sets. Also define sets $R_{\mathrm{odd}}, R_{\mathrm{even}}$ of odd and even vertices respectively on side $R$ of the triangle, and likewise $L_{\mathrm{odd}}, L_{\mathrm{even}},B_{\mathrm{odd}}, B_{\mathrm{even}}$.

Let $|\Phi_{\Gamma}\rangle$ be the $M$-qubit graph state for $\Gamma$ as in Eq.~\eqref{eq:graphstatedefn}. For $b=b_ub_vb_w\in \{0,1\}^3$ define
\[
\mathcal{T}(b)=\left\{z\in\{0,1\}^M: \; \langle z| H^{\otimes M} S_u^{b_u} S_v^{b_v} S_w^{b_w}|\Phi_{\Gamma}\rangle \neq 0\right\}.
\]
 In other words $\mathcal{T}(b)$ describes the set of possible measurement outcomes if the qubits of $|\Phi_{\Gamma}\rangle$  in $L\cup R\cup B$ are measured in the $X$ basis and the qubits $u,v,w$ are measured in the $X$ or $Y$ basis depending on $b$ (as before, $0$ and $1$ bits of $b$ indicate measurements in the $X$ and $Y$ basis respectively).

For any two vertices $j,k\in \Gamma$ write $\mathrm{dist}_{\Gamma}(j,k)$ for the number of edges in the shortest path between them in $\Gamma$. Define
\[
D=\min \left\{\mathrm{dist}_{\Gamma}(u,v),\mathrm{dist}_{\Gamma}(v,w), \mathrm{dist}_{\Gamma}(u,w)\right\}.
\]
Our aim in this Section is to show (following Barrett et al. \cite{barrett2007}) that the relationship between input $b$ and output $z\in \mathcal{T}(b)$ is geometrically nonlocal in the following sense. 
\begin{lemma}
Consider a classical circuit which takes as input a bit string $b=b_ub_vb_w\in\{0,1\}^3$ and a random string $r\in \{0,1\}^\ell$ (drawn from some distribution $\rho$) and outputs $z\in \{0,1\}^M$. Suppose
\begin{equation}
\mathrm{Prob}\left[z\in \mathcal{T}(b)\right]>\frac{7}{8}.
\label{eq:prob1d}
\end{equation}
Then the lightcone of one of the input bits $b_i\in \{b_u,b_v,b_w\}$ contains an output bit $z_q$ such that $\mathrm{dist}_{\Gamma}(i,q)\geq D/2$.
\label{lem:triangle}
\end{lemma}
The authors of Ref.~\cite{barrett2007} showed that the  correlations resulting from measuring the graph state $|\Phi_\Gamma\rangle$
 cannot be reproduced by a ``communication-assisted local hidden variable model''. In Lemma \ref{lem:triangle} we have phrased the result in terms of circuits.

We shall prove the lemma momentarily.  First we show that the set $\mathcal{T}(b)$ of possible measurement outcomes for the 1D graph state satisfies an identity similar to Eq.~\eqref{eq:ghz}. It will be convenient to work with $\pm1$-valued variables defined by $m_j=(-1)^{z_j}$ for $j\in\{1,2,\ldots,M\}$. Define the following products:
\begin{equation}
m_L=\prod_{j\in L_{odd}} m_j \qquad m_R=\prod_{j\in R_{odd}} m_j \qquad  m_B=\prod_{j\in B_{odd}} m_j \qquad m_{E}=\prod_{j\in R_{\mathrm{even}}\cup L_{\mathrm{even}}\cup B_{\mathrm{even}}}m_j.
\label{eq:zprod1}
\end{equation}
\begin{claim}
Let $b=b_ub_vb_w\in \{0,1\}^3$ and suppose $z\in \mathcal{T}(b)$. Then $m_Rm_Bm_L=1$. Moreover, if $b_u\oplus b_v\oplus b_w=0$ then
\begin{equation}
i^{b_u+b_v+b_w}m_um_vm_w m_E m_R^{b_u}m_B^{b_v}m_L^{b_w}=1.
\label{eq:zprod}
\end{equation}
\label{claim:zclass}
\end{claim}
\begin{proof}
For any subset $\mathcal{I}\subseteq [M]$ define operators
\[
X(\mathcal{I})=\prod_{j\in \mathcal{I}} X_j \quad \mbox{and} \quad
g(\mathcal{I})=\prod_{j\in \mathcal{I}} g_j.
\]
Recall that $g_j$ are stabilizers of the graph state $|\Phi_\Gamma\rangle$
defined in Eq.~\eqref{eq:stabgroup} such that  $g_j|\Phi_\Gamma\rangle=|\Phi_\Gamma\rangle$
for all $j$. First note that the operator $X(R_{\mathrm{odd}}\cup L_{\mathrm{odd}}\cup B_{\mathrm{odd}})$ is in the stabilizer group 
of $|\Phi_{\Gamma}\rangle$. Indeed, we have 
\[
X(R_{\mathrm{odd}}\cup L_{\mathrm{odd}}\cup B_{\mathrm{odd}})=
g(R_{\mathrm{odd}}\cup L_{\mathrm{odd}}\cup B_{\mathrm{odd}}).
\]
Accordingly,  $|\Phi_{\Gamma}\rangle$ is in the $+1$ eigenspace of this operator. 
Therefore a measurement of each qubit in $R_{\mathrm{odd}}\cup L_{\mathrm{odd}}\cup B_{\mathrm{odd}}$
in the $X$ basis will result in outcomes $m_R$, $m_L$, $m_B$ satisfying $m_Rm_Bm_L=1$ as claimed. The four cases of Eq.~\eqref{eq:zprod} arise in the same way from the following elements of the
stabilizer group of $|\Phi_\Gamma\rangle$:
\begin{align*}
&\text{($b_ub_vb_w=000$)}\qquad X_u X_v X_w \cdot X\big(R_{\mathrm{even}}\cup L_{\mathrm{even}}\cup 
B_{\mathrm{even}}\big) = g(\{uvw\}\cup R_{\mathrm{even}}\cup L_{\mathrm{even}}\cup 
B_{\mathrm{even}}) \\
&\text{ ($b_ub_vb_w=110$)}\qquad  -Y_uY_v X_w \cdot X\big( R\cup B\cup L_{\mathrm{even}}\big)
=g(\{uvw\} \cup R\cup B \cup L_{\mathrm{even}})\\
&\text{ ($b_ub_vb_w=101$)}\qquad -Y_uX_vY_w  \cdot X\big( R\cup L\cup B_{\mathrm{even}}\big) =
g(\{uvw\} \cup R\cup L\cup B_{\mathrm{even}}) \\
&\text{ ($b_ub_vb_w=011$)}\qquad -X_uY_vY_w \cdot X\big( B\cup L\cup R_{\mathrm{even}}\big)=
g(\{uvw\} \cup B\cup L\cup R_{\mathrm{even}}).
\end{align*}
\end{proof}
\begin{proof}[Proof of Lemma \ref{lem:triangle}]
To reach a contradiction let us suppose that the hypotheses of the lemma are satisfied but the conclusion does not hold. That is, let $\mathcal{C}$ be a classical circuit satisfying Eq.~\eqref{eq:prob1d} and suppose that the lightcone $L_{\mathcal{C}}(b_u)$ only includes output bits $z_j$ where $\mathrm{dist}_\Gamma (u,j)\leq D/2-1$ (and likewise for $b_v$ and $b_w$). Therefore each output bit $z_j$ only depends on the random string $r$ as well as the nearest input bit $b_u,b_v$ or $b_w$ (if $z_j$ is equidistant to two of them it depends on neither).

Write $z=F(b,r)$ for the function which is computed by the circuit $\mathcal{C}$. Below we show that for each $r$ there exists a string $b\in \{0,1\}^3$ such that $F(b,r)\notin \mathcal{T}(b)$. This implies that when $r$ is chosen at random from some distribution $\rho$ we have
\[
\frac{1}{8} \sum_{b\in \{0,1\}^3} \mathrm{Prob}_{\rho}\left[F(b,r)\in \mathcal{T}(b)\right] =
\frac{1}{8} \cdot \mathbb{E}_{\rho} \bigg[ \# \left\{ b\in \{0,1\}^3  \, : \, F(b,r)\in \mathcal{T}(b) \right\} \bigg] \leq \frac{7}{8}.
\]
which shows that $\mathrm{Prob}_{\rho}\left[F(b,r)\in \mathcal{T}(b)\right]\leq 7/8$ for some $b$. Thus we arrive at a contradiction, which is sufficient to prove the Lemma.

It remains to show that for each $r$ there exists a $b$ such that $F(b,r)\notin \mathcal{T}(b)$. So let $r$ be fixed and consider $z=F(b,r)$ as a function of $b$. Let $m_j=(-1)^{z_j}$ and consider the products defined in Eq.~\eqref{eq:zprod1} (as a function of $b$). Suppose first that $m_Rm_Bm_L=-1$ for some $b_0\in \{0,1\}^3$. Then by Claim \ref{claim:zclass}, $F(b_0,r)\notin \mathcal{T}(b_0)$ and we are done. Next suppose that $m_Rm_Bm_L=1$ for all $b\in \{0,1\}^3$.  Since each output bit $z_j$ is a function only of the nearest input bit $b_u,b_v,b_w$ and we are considering products of values $(-1)^{z_j}$, there exist affine boolean functions $e,f,g,h:\{0,1\}^3\rightarrow \{0,1\}$ such that 
\[
m_u m_v m_w m_E=(-1)^{e(b)} \qquad m_R=(-1)^{f(b)} \qquad m_B=(-1)^{g(b)} \qquad m_L=(-1)^{h(b)}
\]
and such that $f(b)$ does not depend on $b_u$, $g(b)$ does not depend on $b_v$, $h(b)$ does not depend on $b_w$, and $f(b)\oplus g(b)\oplus h(b)=0$. Note that 
\begin{equation}
i^{b_u+b_v+b_w}m_um_vm_w m_E m_R^{b_u}m_B^{b_v}m_L^{b_w}=i^{b_u+b_v+b_w}(-1)^{e(b)+f(b)b_u+g(b)b_v+h(b)b_w}.
\label{eq:zfg}
\end{equation}
The following Claim implies that Eq.~\eqref{eq:zfg} is not equal to $+1$ for all bit strings $b$ with even hamming weight. Applying Claim \ref{claim:zclass} we see that this implies that $F(b,r)\notin \mathcal{T}(b)$ for some (even hamming weight) string $b$, completing the proof.
\begin{claim}
Suppose $e,f,g,h:\{0,1\}^3\rightarrow \{0,1\}$ are affine boolean functions.  Write $x=x_1x_2x_3\in \{0,1\}^3$. Suppose $f(x)$ does not depend on $x_1$, $g(x)$ does not depend on $x_2$, $h(x)$ does not depend on $x_3$, and that $f(x)\oplus g(x)\oplus h(x)$ is independent of $x$. Then 
\[
\sum_{x_1\oplus x_2\oplus x_3=0} i^{x_1+x_2+x_3}(-1)^{e(x)+f(x)x_1+g(x)x_2+h(x)x_3} \leq 2.
\]
\label{claim:affine}
\end{claim}
A proof of Claim \ref{claim:affine} is provided in Appendix \ref{app:affine}.
\end{proof}

\subsection{Proof of Theorem \ref{thm:lowerbound}\label{sec:2dnonlocal}} 
\begin{proof}
Let $\mathcal{C}\equiv \mathcal{C}_N$ be a classical probabilistic circuit of fan-in $\leq K$ which solves the 2D Hidden Linear Function problem with probability $>7/8$ on all instances of size $N$. That is, $\mathcal{C}$ takes as input vectors $A,b$ as in Eq.~\eqref{q2D} as well as a random bit string $r$ drawn from some (arbitrary) distribution. The output of $\mathcal{C}$ is a bit string $z$ which is a solution to the given instance with probability $>7/8$. 

We suppose that the depth $d$ of $\mathcal{C}$ satisfies
\begin{equation}
d< \frac{1}{8}\frac{\log(N)}{\log(K)}.
\label{eq:depth}
\end{equation}
Below we prove that for all sufficiently large $N$ (i.e., larger than some universal constant) this leads to a contradiction. 

Suppose $j\in V$ is a vertex of the $N\times N$ grid $G=(V,E)$. Let $\bbox{j}\subseteq V$ be a square box of size $\lfloor N^{1/2}\rfloor \times \lfloor N^{1/2}\rfloor$ centered at vertex $j$. Each box defines a subset of output variables $z_k$ contained in this box.
Choose square-shaped regions $\mathcal{U},\mathcal{V},\mathcal{W}\subseteq V$ as shown in Figure~\ref{fig:1}. Let $V_{\mathrm{even}}\subset V$ denote the set of vertices on the even sublattice of the grid. In other words $V_{\mathrm{even}}$ contains all vertices with even horizontal and vertical coordinates.  

Combining Eqs.~(\ref{eq:lightconesize},\ref{eq:depth}) we get
\begin{equation}
|L_{\mathcal{C}}(z_i)|\leq K^d<N^{\frac{1}{8}} \qquad i\in V.
\label{eq:Kd}
\end{equation}
This shows that all output bits have ``small" lightcones. Next we shall identify large sets of input bits which also have small lightcones.

For each region $\mathcal{R}\in \{\mathcal{U},\mathcal{V},\mathcal{W}\}$ define sets of good and bad vertices
\begin{align}
\mathrm{Good}(\mathcal{R})&=\{v\in \mathcal{R}\cap V_{\mathrm{even}}: |L_\mathcal{C}(b_v)|\leq N^\frac{1}{4}\} \label{good} \\
\mathrm{Bad}(\mathcal{R})&=\{\mathcal{R}\cap V_{\mathrm{even}}\}\setminus \mathrm{Good}(\mathcal{R}).
\end{align}

\begin{claim}
For $\mathcal{R}\in \{\mathcal{U},\mathcal{V},\mathcal{W}\}$ we have
\begin{equation}
|\mathrm{Good}(\mathcal{R})|=\Omega(N^2).
\label{eq:sizegood}
\end{equation}
\end{claim}
\begin{proof}
Define a bipartite graph with one side of the partition labeled by input bits $b_v$ with $v\in \mathcal{R}\cap V_{\mathrm{even}}$ and the other side labeled by outputs $z_j$ with $j\in V$. An edge between $z_j$ and $b_v$ is present iff $b_v\in L_{\mathcal{C}}(z_j)$. The total number of edges $J$ in this graph satisfies
\[
|\mathrm{Bad}(\mathcal{R})|N^\frac{1}{4}\leq J\leq |V|\cdot \max_{i\in V} |L_\mathcal{C}(z_i)|\leq N^{\frac{17}{8}},
\]
where we used $|V|=N^2$ and Eq.~\eqref{eq:Kd}. Rearranging gives $|\mathrm{Bad}(\mathcal{R})|\leq N^\frac{15}{8}$. Since $|\mathcal{R}\cap V_{\mathrm{even}}|=\Theta(N^2)$ we get
\[
|\mathrm{Good}(\mathcal{R})|= |\mathcal{R}\cap V_{\mathrm{even}}|-|\mathrm{Bad}(\mathcal{R})|=\Omega(N^2).
\]
\end{proof}

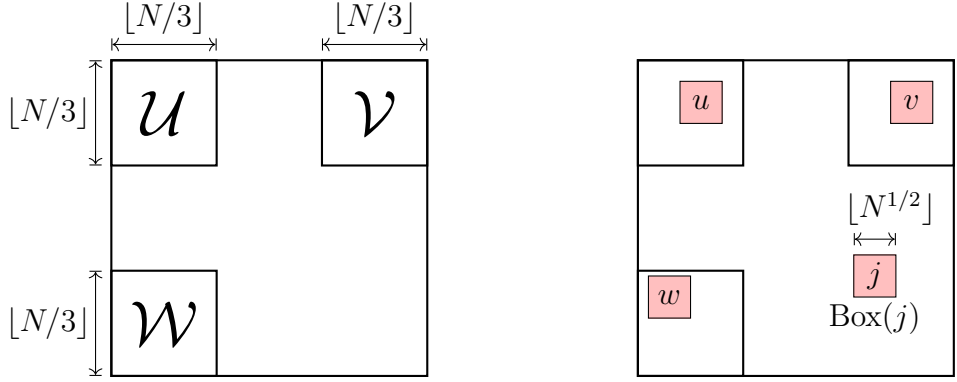
\begin{figure}
\centering
\begin{tikzpicture}[scale=0.7]
\draw[thick] (0,0) rectangle (6,6);
\draw[thick] (0,0) rectangle (2,2);
\draw[thick] (0,4) rectangle (2,6);
\draw[thick] (4,4) rectangle (6,6);
\node at (1,1){\scalebox{2}{$\mathcal{W}$}};
\node at (5,5){\scalebox{2}{$\mathcal{V}$}};
\node at (1,5){\scalebox{2}{$\mathcal{U}$}};
\draw[ |<->|]  (0,6.3)--(2,6.3);
\draw[|<->|]  (4,6.3)--(6,6.3);
\draw[|<->|]  (-0.3,0)--(-0.3,2);
\draw[|<->|]  (-0.3,4)--(-0.3,6);
\node at (1,6.8){$\lfloor N/3\rfloor $};
\node at (5,6.8){$\lfloor N/3\rfloor $};
\node at (-1.2,1){$\lfloor N/3\rfloor $};
\node at (-1.2,5){$\lfloor N/3\rfloor $};

\begin{scope}[shift={(10,0)}]
\draw[thick] (0,0) rectangle (6,6);
\draw[thick] (0,0) rectangle (2,2);
\draw[thick] (0,4) rectangle (2,6);
\draw[thick] (4,4) rectangle (6,6);

\draw [fill=pink] (0.2,1.1) rectangle (1,1.9);
\node at (0.6,1.5){$w$};

\draw [fill=pink] (0.8,4.8) rectangle (1.6,5.6);
\node at (1.2,5.2){$u$};

\draw [fill=pink] (4.8,4.8) rectangle (5.6,5.6);
\node at (5.2,5.2){$v$};

\draw[|<->|] (4.1,2.6)--(4.9,2.6);
\node at (4.8, 3.2){$\lfloor N^{1/2}\rfloor$};
\draw [fill=pink] (4.1,1.5) rectangle (4.9,2.3);
\node at (4.5,1.9){$j$};
\node at (4.5, 1.1){$\bbox{j}$};
\end{scope}
\end{tikzpicture}
\caption{Left: definition of the regions $\mathcal{U},\mathcal{V},\mathcal{W}$. Right: definition of $\bbox{j}$
and a possible choice of  vertices $u,v,w$. }
\label{fig:1}
\end{figure}

\begin{claim}
For all large enough $N$ one can choose a triple of vertices $u,v,w$ such that $u\in \mathrm{Good}(\mathcal{U})$, $v\in \mathrm{Good}(\mathcal{V})$, $w\in \mathrm{Good}(\mathcal{W})$ and
\begin{equation}
\label{c1}
\bbox{u}\subseteq \mathcal{U}, \quad \bbox{v}\subseteq \mathcal{V}, \quad \bbox{w}\subseteq \mathcal{W},
\end{equation}
\begin{equation}
\label{c2}
L_\mathcal{C} (b_u)\cap \bbox{v} =\emptyset, \quad L_{\mathcal{C}}(b_u)\cap \bbox{w} =\emptyset
\end{equation}
\begin{equation}
\label{c3}
L_\mathcal{C}(b_v)\cap \bbox{u} =\emptyset, \quad L_\mathcal{C}(b_v)\cap \bbox{w} =\emptyset
\end{equation}
\begin{equation}
\label{c4}
L_\mathcal{C}(b_w)\cap \bbox{u} =\emptyset, \quad L_\mathcal{C}(b_w)\cap \bbox{v} =\emptyset.
\end{equation}
\label{claim:triple}
\end{claim}
\begin{proof}
Since each vertex in the grid belongs to at most $N$ boxes, we infer that a given
lightcone $L_\mathcal{C}(b_u)$ with $u\in \mathrm{Good}(\mathcal{U})$ can intersect with at most $N |L_{\mathcal{C}}(b_u)|\leq N^{\frac{5}{4}}$ boxes. Here we used the fact that $|L_{\mathcal{C}}(b_u)|\leq N^{\frac{1}{4}}$ for all $u\in \mathrm{Good}(\mathcal{U})$ by definition. The total number of vertices $v\in \mathrm{Good}(\mathcal{V})$ such that $\bbox{v}\subseteq \mathcal{V}$ is $\Omega (N^2)$, which follows from Eq.~\eqref{eq:sizegood} (and since the number of vertices $q\in \mathcal{V}$ with $\bbox{q}\not\subseteq \mathcal{V}$ is $o(N^2)$ as any such vertex $q$ must lie near the boundary of region $\mathcal{V}$). Thus if $u,v,w$ are picked  uniformly at random from the sets $\mathrm{Good}(\mathcal{U})$, $\mathrm{Good}(\mathcal{V})$ and $\mathrm{Good}(\mathcal{W})$ respectively subject to Eq.~(\ref{c1}) then 
\[
\mathrm{Prob}[L_\mathcal{C}(b_u) \cap \bbox{v}\ne \emptyset] \le O\left(\frac{N^{\frac{5}{4}}}{N^{2}}\right)<\frac16
\] 
for large enough $N$. A similar bound applies to the five other combinations of vertices that appear in Eqs.~(\ref{c2},\ref{c3},\ref{c4}).
By the union bound, there exists at least one choice of $u,v,w$ that satisfies all
conditions Eqs.~(\ref{c1},\ref{c2},\ref{c3},\ref{c4}).
\end{proof}
Below we consider cycles $\Gamma$ that are subgraphs of the grid $G$. 
\begin{claim}The following holds for all sufficiently large $N$. Fix some triple of vertices $u\in \mathrm{Good}(\mathcal{U})$, $v\in \mathrm{Good}(\mathcal{V})$, $w\in \mathrm{Good}(\mathcal{W})$
satisfying Eqs.~(\ref{c1},\ref{c2},\ref{c3},\ref{c4}). Then there exists a cycle $\Gamma$ containing
$u,v,w$ such that the lightcones $L_\mathcal{C}(b_u)$, $L_\mathcal{C}(b_v)$, $L_\mathcal{C}(b_w)$
contain no vertices of $\Gamma$ lying outside 
of $\bbox{u}\cup \bbox{v}\cup \bbox{w}$.
\label{claim:gamma}
\end{claim}

\begin{figure}[h]
\centering
\begin{tikzpicture}
\draw[thick] (0,0) rectangle (6,6);

\draw [fill=pink] (0.2,1.1) rectangle (1,1.9);
\node at (0.6,1.5){\scalebox{2}{$w$}};

\draw [fill=pink] (0.8,4.8) rectangle (1.6,5.6);
\node at (1.2,5.2){\scalebox{2}{$u$}};

\draw [fill=pink] (4.8,4.8) rectangle (5.6,5.6);
\node at (5.2,5.2){\scalebox{2}{$v$}};

\draw[thick,dashed] (1.6,5.4)--(4.8,5.4);
\draw[thick,dashed] (1.6,5.2)--(4.8,5.2);
\draw[thick,dashed] (1.6,5.0)--(4.8,5.0);

\draw[thick,dashed] (5,4.8)--(5,1.7)--(1,1.7);
\draw[thick,dashed] (5.2,4.8)--(5.2,1.5)--(1,1.5);
\draw[thick,dashed] (5.4,4.8)--(5.4,1.3)--(1,1.3);

\draw[thick,dashed] (1,4.8)--(1,3.2)--(0.4,3.2)--(0.4,1.9);
\draw[thick,dashed] (1.2,4.8)--(1.2,3)--(0.6,3)--(0.6,1.9);
\draw[thick,dashed] (1.4,4.8)--(1.4,2.8)--(0.8,2.8)--(0.8,1.9);

\node at (3.2,4.6){\scalebox{0.8}{$\gamma(u,v)$}};
\node at (4.4,2.7){\scalebox{0.8}{$\gamma(v,w)$}};

\node at (2, 3.3){\scalebox{0.8}{$\gamma(u,w)$}};

\end{tikzpicture}
\caption{Pairwise disjoint paths $\gamma$ connecting the boxes
$\bbox{u}$, $\bbox{v}$, $\bbox{w}$. The number of paths connecting each pair
of boxes is $\lfloor N^{1/2}\rfloor$. }
\label{fig:2}
\end{figure}
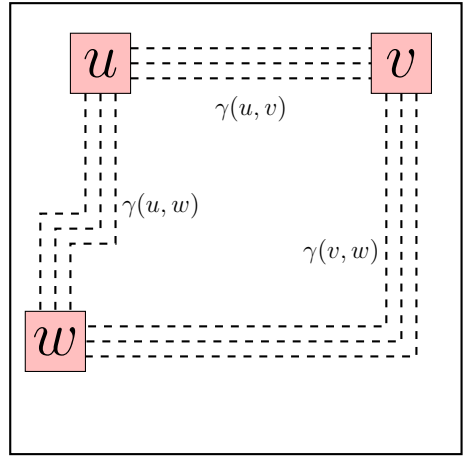

\begin{proof}
Indeed, since each box  has size $\lfloor N^{1/2}\rfloor \times \lfloor N^{1/2}\rfloor$, one can choose
$\lfloor N^{1/2}\rfloor$ pairwise disjoint paths $\gamma$ that connect 
any pair of boxes $\bbox{u}$, $\bbox{v}$, $\bbox{w}$, see Figure~\ref{fig:2}.
Let $\gamma(a,b)$ be a path connecting $\bbox{a}$ and $\bbox{b}$,
where $a\ne b \in \{u,v,w\}$.
Any triple of paths $\gamma(u,v)$, $\gamma(v,w)$, $\gamma(u,w)$ can be completed
to a cycle $\Gamma$ that contains $u,v,w$  by adding the missing segments of the cycle  inside 
the boxes $\bbox{u}$, $\bbox{v}$, $\bbox{w}$.
Since $L_\mathcal{C}(b_u)$ has size at most $N^{\frac{1}{4}}$ (recall that $u$ is a good vertex)
 and each vertex $q\in V$ belongs to at most one
path $\gamma$, we infer that $L_\mathcal{C}(b_u)$ intersects with at most $N^{\frac{1}{4}}$ paths $\gamma$.
Thus if we pick the path $\gamma(u,v)$ uniformly at random
among all $\lfloor N^{1/2}\rfloor$ possible choices then
\[
\mathrm{Prob}[L_\mathcal{C}(b_u)\cap \gamma(u,v)\ne \emptyset] \le \frac{N^{\frac{1}{4}}}{\lfloor N^{1/2}\rfloor}< \frac19
\]
for large enough $N$. The same bound applies to eight remaining  combinations of a lightcone
$L_\mathcal{C}(b_u)$, $L_\mathcal{C}(b_v)$, $L_\mathcal{C}(b_w)$ and a path $\gamma$.
By the union bound, there exists at least one triple of paths $\gamma(u,v)$, $\gamma(v,w)$, $\gamma(u,w)$
that do not intersect with $L_\mathcal{C}(b_u)$, $L_\mathcal{C}(b_v)$, $L_\mathcal{C}(b_w)$.
This gives the desired cycle $\Gamma$.
\end{proof}
Let $u,v,w$ and $\Gamma$ be chosen as described in Claim \ref{claim:gamma}. Let $M$ be the number of vertices in $\Gamma$. Recall that $u,v,w\in V_{\mathrm{even}}$ and therefore all pairwise distances between them along $\Gamma$ are even. Consider the subset of instances $(A,b)$ of the 2D Hidden Linear Function problem where
\begin{equation}
A_e=\begin{cases} 1 & \text{if $e$ is an edge of $\Gamma$}\\0 & \text{otherwise.}\end{cases} \qquad \quad \text{and} \qquad  \quad b_j=0 \quad \text{if} \quad j\in V\setminus \{u,v,w\}.
\label{eq:AB}
\end{equation}
There are $8$ such instances corresponding to choices of input bits $b_u,b_v,b_w\in \{0,1\}$. By fixing inputs to the circuit $\mathcal{C}$ in this way and looking only at output bits $z_j$ with $j\in \Gamma$ we obtain a classical circuit $\mathcal{D}$ which takes a three-bit string  $b_{u}b_{v}b_{w}\in \{0,1\}^3$ and a random string $r$ as input and outputs $z_\Gamma\in \{0,1\}^M$. For any input bit $b_i\in \{b_u,b_v,b_w\}$ we have 
\begin{equation}
L_{\mathcal{D}}(b_i)\subseteq L_{\mathcal{C}}(b_i)
\label{eq:lcones}
\end{equation}
since any pair of input/output variables which are correlated in $\mathcal{D}$ are by definition also correlated in $\mathcal{C}$. Our assumption that $\mathcal{C}$ solves the 2D Hidden Linear Function problem with probability greater than $7/8$ implies 
\begin{equation}
\mathrm{Prob}\left[\langle z_\Gamma| H^{\otimes M}S_u^{b_u}S_{v}^{b_v} S_w^{b_w}|\Phi_{\Gamma}\rangle \neq 0\right]>\frac{7}{8}.
\label{eq:assumptionimplies}
\end{equation}
To see this, recall from Lemma \ref{lem:pz} that the circuit from Fig.~\ref{fig:hlf} produces a uniformly random solution to the 2D Hidden Linear Function problem, and that the state of output qubits in $\Gamma$ after applying the circuit to inputs $A,b$ as in Eq.~\eqref{eq:AB} is 
\[
H^{\otimes M}S_u^{b_u}S_{v}^{b_v} S_w^{b_w}|\Phi_{\Gamma}\rangle.
\]
Using Eq.~\eqref{eq:assumptionimplies} and applying Lemma \ref{lem:triangle} with the cycle $\Gamma$ constructed above we infer that the lightcone $L_{\mathcal{D}}(b_i)$ of one of the input bits
$b_i\in \{b_u,b_v,b_w\}$ contains at least one output bit $z_j$ such that $j\in \Gamma$ and the distance
between $i$ and $j$ along the cycle $\Gamma$ is $\Omega(N)$.  By Eq.~\eqref{eq:lcones} the same is true for the lightcone $L_{\mathcal{C}}(b_i)$. For all sufficiently large $N$ this contradicts Claims \ref{claim:triple},\ref{claim:gamma}. Indeed, by Claim \ref{claim:gamma}, the vertex $j \in L_C(b_i)\cap \Gamma$ must lie in one of $\bbox{u}, \bbox{v}$ or $\bbox{w}$, and since $L_\mathcal{C}(b_i)$ has no intersection with $\bbox{k}$  ($i\neq k$) by Claim \ref{claim:triple}, this implies that $j\in \bbox{i}$. But the distance from $i$ to any vertex inside $\bbox{i}$ is $\leq N^{1/2}$. We conclude that Eq.~\eqref{eq:depth} is false for all sufficiently large $N$.
\end{proof}

\section{Conclusions and open problems \label{sec:conc}}
We have shown that shallow quantum circuits are more powerful than their classical counterparts. Our work raises several questions:

Does the 2D Hidden Linear Function problem resist simulation by more powerful classical circuit families such as $\AC^{0}$ (constant-depth circuits with unbounded fan-in)? The constant $7/8$ appearing in Theorem \ref{thm:lowerbound} is unlikely to be optimal--can it be replaced by a vanishing function of $N$? A recursive variant of the Bernstein-Vazirani problem is known to provide a superpolynomial speedup in query complexity \cite{BV}. Is there any use in defining a recursive variant of the Hidden Linear Function problem? 

A challenging open question is to establish a separation between the power of quantum and classical circuits for sampling problems. Let $U_n$ be a (unitary implemented by a) $n$-qubit quantum circuit and consider the output probability distribution 
\[
\rho_n(z)=|\langle z|U_n|0^n\rangle|^2.
\]
Does there exist a constant-depth family of quantum circuits $\{U_n\}$ which sample distributions $\rho_n$ that cannot be sampled by constant-depth classical circuits? A powerful tool that might be used to address this question is given in Ref. \cite{viola2014extractors}. In particular, the author shows that any distribution over $n$-bit strings with linear min-entropy (i.e., $\tilde{\Theta}(n)$) which is generated by applying a constant-depth classical circuit to a uniformly random input string can be expressed as a convex combination of simple distributions (``bit-block sources") with linear min-entropy. We do not know if the distributions sampled by constant-depth quantum circuits have this property.


\section*{Acknowledgments}
SB and DG acknowledge support from the IBM Research Frontiers Institute. RK is supported by the Technische Universit\"at M\"unchen -- Institute for Advanced Study, funded by the German Excellence Initiative and the European Union Seventh Framework Programme under grant agreement no.~291763.

\appendix
\section{Proof of Claim \ref{claim:partialfourier2} \label{app:correctness}}

\begin{proof}
Define a function $J\, : \, \FF_2^n\times \FF_2^n \to \ZZ_4$
\begin{equation}
\label{Jxy}
J(x,y)=q(x\oplus y) - q(x) -q(y), \qquad x,y\in \FF_2^n.
\end{equation}

Let us show that $J(x,y)$ is 
 a bilinear form, that is, there exists a symmetric binary matrix $B$ such that 
\begin{equation}
\label{J}
J(x,y)=2 \sum_{\alpha,\beta=1}^n B_{\alpha,\beta} x_\alpha y_\beta
\end{equation}
for all $x,y$. Indeed,  a direct inspection of Eq.~\eqref{q} shows that $q(x)$ obeys the following 
identity\footnote{By linearity, it suffices to check this identity for the special cases
$q(x)=2x_\alpha x_\beta$ and $q(x)=x_\alpha$.}:
\begin{equation}
\label{q'''=0}
q(x\oplus y \oplus z) -q(x\oplus y) -q(y\oplus z) - q(z\oplus x) + q(x) + q(y) + q(z)=0
\end{equation}
for all $x,y,z$. 
The above expression can be viewed as as a discrete version of the third derivative of $q$ which vanishes
because $q$ is a quadratic form. 
Combining Eqs.~(\ref{Jxy},\ref{q'''=0}) one gets
\begin{equation}
\label{bilinear}
J(x\oplus y,z)=J(x,z) + J(y,z) 
\end{equation}
for all $x,y,z$.  Choosing $x=y$ gives $0=2J(x,z)$, that is, the function $J(x,z)$ takes 
values $0,2$ and we can write $J(x,y)=2J'(x,y)$, where $J'(x,y)$ takes values in $\ZZ_2$.
By definition, $J'(x,y)=J'(y,x)$ and $J'(x\oplus y,z)=J'(x,z)\oplus J'(y,z)$ for all $x,y,z$.
This is possible only if $J'(x,y)$ is a symmetric bilinear form over the binary field, that is, 
$J'(x,y)=x^T B y \pmod{2}$ for some  symmetric binary matrix $B$.
This proves Eq.~(\ref{J}).

The definition of $\calL_q$ and Eqs.~(\ref{Jxy},\ref{J}) imply that 
\begin{equation}
\label{kerB}
\calL_q=\ker{(B)}
\end{equation}
which gives 
\begin{equation}
\label{B}
\ker(B) \cap \calK=0,
\end{equation}
see Eq.~(\ref{LK}).
We claim that for any  $z\in \FF_2^n$ there exists
a vector $w\in \calK$ such that 
\begin{equation}
\label{tricky1}
z^T x =w^T B x \quad  \mbox{for all} \quad  x\in \calK.
\end{equation}
Indeed, note that $(\calX\cap \calY)^\perp =\calX^\perp + \calY^\perp$
for any linear subspaces $\calX,\calY\subseteq \FF_2^n$.
Let $\im{B}\equiv \mathrm{span}\{Bx\, : \, x\in \FF_2^n\}$.
Using the identity
\[
\ker{(B)}^\perp=\im{B^T}=\im{B}
\]
and  taking the dual of Eq.~(\ref{B}) gives
\begin{equation}
\label{B1}
\im{B}+ \calK^\perp=\FF_2^n
\end{equation}
Choose any vector $z\in \FF_2^n$ and write it as 
\[
z=u\oplus v,\quad \mbox{where} \quad u\in \im{B} \quad \mbox{and} \quad v\in \calK^\perp.
\]
This is always possible due to Eq.~(\ref{B1}). 
Let $u=Bu'$ for some $u'\in \FF_2^n$.
From Eq.~(\ref{LK}) we infer that $u'=w\oplus w'$ for some $w\in \calK$ and $w'\in \calL_q$.
Putting together the above facts we see that any vector $z\in \FF_2^n$ can be written as 
\[
z=Bu'\oplus v = B(w\oplus w') \oplus v,  \quad \mbox{where} \quad w\in \calK, \quad  w'\in \calL_q,
\quad v\in \calK^\perp.
\]
Note that $Bw'=0$ due to Eq.~(\ref{kerB}). Thus
$z^T  x =w^T Bx$ for all $x\in \calK$, as claimed in Eq.~(\ref{tricky1}).
From Eq.~(\ref{tricky1}) one gets
\[
(-1)^{z^T x} \cdot i^{q(x)} = 
(-1)^{w^T B x} \cdot i^{q(x)}=
i^{J(w,x) + q(x)} = i^{q(w\oplus x) -q(w)} \qquad \mbox{for all $x\in \calK$}.
\]
Therefore
\[
\Gamma(\calK,z)= \sum_{x\in \calK}(-1)^{z^T x} \cdot i^{q(x)}=
i^{-q(w)} \cdot \sum_{x\in \calK} i^{q(w\oplus x)}=
i^{-q(w)} \cdot \sum_{x\in \calK} i^{q(x)}.
\]
This shows that the absolute value of $\Gamma(\calK,z)$ does not depend on $z$.
Let $C\equiv |\Gamma(\calK,z)|^2$.
Combining Eqs.~(\ref{GG},\ref{GG1}) and Claim~\ref{claim:partialfourier1}  one gets
\[
1=\sum_{z\in \FF_2^n} p(z) = \frac{ |\calL_q^\perp| \cdot |\calL_q|^2 \cdot C}{4^n}
=\frac{|\calL_q| \cdot C}{2^n},
\]
which proves the claim. 
\end{proof}


\section{Proof of Claim \ref{claim:affine} \label{app:affine}}
\begin{proof}

Write 
\begin{align}
e(x)&=e_0\oplus e_1x_1\oplus e_2x_2\oplus e_3x_3  \quad &e_0,e_1,e_2,e_3\in \{0,1\}\\
f(x)&=f_0\oplus f_2x_2\oplus f_3x_3 \qquad &f_0,f_2,f_3\in \{0,1\}\\
g(x)&=g_0\oplus g_1x_1\oplus g_3x_3  \qquad &g_0,g_1,g_3\in \{0,1\}\\
h(x)&=h_0\oplus h_1x_1\oplus h_2x_2.   \qquad &h_0,h_1,h_2\in \{0,1\}
\end{align}
The fact that $f(x)\oplus g(x)\oplus h(x)$ is a constant function implies 
\begin{equation}
f_2\oplus h_2=f_3\oplus g_3=g_1\oplus h_1=0.
\label{eq:coefs}
\end{equation}
We have
\begin{align}
f(x)x_1&+g(x)x_2+h(x)x_3\nonumber\\
&=f_0x_1+g_0x_2+h_0x_3+(f_2+g_1)x_1x_2+(f_3+h_1)x_1x_3+(g_3+h_2)x_2x_3.
\label{eq:val}
\end{align}
For all $x$ satisfying $x_1\oplus x_2\oplus x_3=0$ we have $x_1x_3=x_1x_2\oplus x_1$ and $x_2x_3=x_1x_2\oplus x_2$. Using this fact and Eqs.\eqref{eq:val}, \eqref{eq:coefs} we get
\[
(-1)^{f(x)x_1+g(x)x_2+h(x)x_3}=(-1)^{(f_0+f_3+h_1)x_1+(g_0+g_3+h_2)x_2+h_0x_3}
\]
whenever $x_1\oplus x_2\oplus x_3=0$. Noting that the exponent on the right hand side is an affine boolean function, and that $e(x)$ is also an affine boolean function we get
\begin{align}
\sum_{x_1\oplus x_2\oplus x_3=0} i^{x_1+x_2+x_3}(-1)&^{e(x)+f(x)x_1+g(x)x_2+h(x)x_3}\nonumber\\
&\leq \max_{w\in \{0,1\}^4} \sum_{x_1\oplus x_2\oplus x_3=0} i^{x_1+x_2+x_3}(-1)^{w_0+w_1x_1+w_2x_2+w_3x_3}\label{eq:lastline}\\
&\leq  2\nonumber
\end{align}
Since each summand in Eq.~\eqref{eq:lastline} is $\pm 1$, the last line is equivalent to the statement that the sum is strictly less than $4$, which follows from the fact that the following system of equations over $\mathbb{F}_2$ has no solution:
\begin{equation}
w_0=0 \qquad w_1\oplus w_2=1 \qquad w_1\oplus w_3=1 \qquad w_2\oplus w_3=1.
\label{eq:weqs}
\end{equation}
(Note that Eq.~\eqref{eq:weqs} would be necessary for Eq.~\eqref{eq:lastline} to be equal to $4$, as can be seen by considering $x=x_1 x_2 x_3\in \{000,110,101,011\}$.)
\end{proof}

\bibliographystyle{unsrt}
\bibliography{constdepth}

\begin{thebibliography}{10}

\bibitem{csanky1976fast}
Laszlo Csanky.
\newblock Fast parallel matrix inversion algorithms.
\newblock {\em SIAM Journal on Computing}, 5(4):618--623, 1976.

\bibitem{borodin1982fast}
Allan Borodin, Joachim Von Zur~Gathen, and John Hopcroft.
\newblock Fast parallel matrix and gcd computations.
\newblock {\em Information and Control}, 52(3):241--256, 1982.

\bibitem{arora2009computational}
Sanjeev Arora and Boaz Barak.
\newblock {\em Computational complexity: a modern approach}.
\newblock Cambridge University Press, 2009.

\bibitem{terhal2002adaptive}
Barbara~M Terhal and David~P DiVincenzo.
\newblock Adaptive quantum computation, constant depth quantum circuits and
  {A}rthur-{M}erlin games.
\newblock {\em Quant. Inf. Comp.}, 4(2):134--145, 2004.

\bibitem{bermejo2017architectures}
Juan Bermejo-Vega, Dominik Hangleiter, Martin Schwarz, Robert Raussendorf, and
  Jens Eisert.
\newblock Architectures for quantum simulation showing quantum supremacy.
\newblock {\em arXiv preprint arXiv:1703.00466}, 2017.

\bibitem{bremner2016average}
Michael~J Bremner, Ashley Montanaro, and Dan~J Shepherd.
\newblock Average-case complexity versus approximate simulation of commuting
  quantum computations.
\newblock {\em Physical Review Letters}, 117(8):080501, 2016.

\bibitem{bremner2016achieving}
Michael~J Bremner, Ashley Montanaro, and Dan~J Shepherd.
\newblock Achieving quantum supremacy with sparse and noisy commuting quantum
  computations.
\newblock {\em arXiv preprint arXiv:1610.01808}, 2016.

\bibitem{farhi2016quantum}
Edward Farhi and Aram~W Harrow.
\newblock Quantum supremacy through the quantum approximate optimization
  algorithm.
\newblock {\em arXiv preprint arXiv:1602.07674}, 2016.

\bibitem{gao2016quantum}
Xun Gao, Sheng-Tao Wang, and Lu-Ming Duan.
\newblock Quantum supremacy for simulating a translation-invariant {I}sing spin
  model.
\newblock {\em Physical Review Letters}, 118:040502, 2017.

\bibitem{cleve2000fast}
Richard Cleve and John Watrous.
\newblock Fast parallel circuits for the quantum {F}ourier transform.
\newblock In {\em Foundations of Computer Science, 2000. Proceedings. 41st
  Annual Symposium on}, pages 526--536. IEEE, 2000.

\bibitem{temme2016error}
Kristan Temme, Sergey Bravyi, and Jay~M Gambetta.
\newblock Error mitigation for short depth quantum circuits.
\newblock {\em arXiv preprint arXiv:1612.02058}, 2016.

\bibitem{li2016efficient}
Ying Li and Simon~C Benjamin.
\newblock Efficient variational quantum simulator incorporating active error
  minimisation.
\newblock {\em arXiv preprint arXiv:1611.09301}, 2016.

\bibitem{boixo2016characterizing}
Sergio Boixo, Sergei~V Isakov, Vadim~N Smelyanskiy, Ryan Babbush, Nan Ding,
  Zhang Jiang, John~M Martinis, and Hartmut Neven.
\newblock Characterizing quantum supremacy in near-term devices.
\newblock {\em arXiv preprint arXiv:1608.00263}, 2016.

\bibitem{gottesman2013overhead}
Daniel Gottesman.
\newblock What is the overhead required for fault tolerance?
\newblock {\em arXiv preprint arXiv:1310.2984}, 2013.

\bibitem{bombin2015gauge}
H{\'e}ctor Bomb{\'\i}n.
\newblock Gauge color codes: optimal transversal gates and gauge fixing in
  topological stabilizer codes.
\newblock {\em New Journal of Physics}, 17(8):083002, 2015.

\bibitem{pastawski2009long}
Fernando Pastawski, Alastair Kay, Norbert Schuch, and Ignacio Cirac.
\newblock How long can a quantum memory withstand depolarizing noise?
\newblock {\em Physical Review Letters}, 103(8):080501, 2009.

\bibitem{BV}
Ethan Bernstein and Umesh Vazirani.
\newblock Quantum complexity theory.
\newblock {\em SIAM Journal on Computing}, 26(5):1411--1473, 1997.

\bibitem{mermin}
David Mermin.
\newblock Extreme quantum entanglement in a superposition of macroscopically
  distinct states.
\newblock {\em Physical Review Letters}, 65(15):1838, 1990.

\bibitem{ghz}
Daniel~M Greenberger, Michael~A Horne, Abner Shimony, and Anton Zeilinger.
\newblock Bell’s theorem without inequalities.
\newblock {\em American Journal of Physics}, 58(12):1131--1143, 1990.

\bibitem{barrett2007}
Jonathan Barrett, Carlton~M Caves, Bryan Eastin, Matthew~B Elliott, and Stefano
  Pironio.
\newblock Modeling {P}auli measurements on graph states with nearest-neighbor
  classical communication.
\newblock {\em Physical Review A}, 75(1):012103, 2007.

\bibitem{raussendorf2001one}
Robert Raussendorf and Hans~J Briegel.
\newblock A one-way quantum computer.
\newblock {\em Physical Review Letters}, 86(22):5188, 2001.

\bibitem{briegel2001persistent}
Hans~J Briegel and Robert Raussendorf.
\newblock Persistent entanglement in arrays of interacting particles.
\newblock {\em Physical Review Letters}, 86(5):910, 2001.

\bibitem{deutsch1992rapid}
David Deutsch and Richard Jozsa.
\newblock Rapid solution of problems by quantum computation.
\newblock In {\em Proceedings of the Royal Society of London A: Mathematical,
  Physical and Engineering Sciences}, volume 439, pages 553--558. The Royal
  Society, 1992.

\bibitem{atici2007quantum}
Alp At{\i}c{\i} and Rocco~A Servedio.
\newblock Quantum algorithms for learning and testing juntas.
\newblock {\em Quantum Information Processing}, 6(5):323--348, 2007.

\bibitem{aaronson2010bqp}
Scott Aaronson.
\newblock {BQP} and the polynomial hierarchy.
\newblock In {\em Proceedings of the forty-second ACM symposium on Theory of
  computing}, pages 141--150. ACM, 2010.

\bibitem{gavinsky2011quantum}
Dmitry Gavinsky, Martin Roetteler, and J{\'e}r{\'e}mie Roland.
\newblock Quantum algorithm for the boolean hidden shift problem.
\newblock In {\em International Computing and Combinatorics Conference}, pages
  158--167. Springer, 2011.

\bibitem{cross2015quantum}
Andrew~W Cross, Graeme Smith, and John~A Smolin.
\newblock Quantum learning robust against noise.
\newblock {\em Physical Review A}, 92(1):012327, 2015.

\bibitem{rotteler2009quantum}
Martin R{\"o}tteler.
\newblock Quantum algorithms to solve the hidden shift problem for quadratics
  and for functions of large gowers norm.
\newblock In {\em International Symposium on Mathematical Foundations of
  Computer Science}, pages 663--674. Springer, 2009.

\bibitem{Aaronson16}
Scott Aaronson and Lijie Chen.
\newblock Complexity-theoretic foundations of quantum supremacy experiments.
\newblock {\em arXiv preprint arXiv:1612.05903}, 2016.

\bibitem{moore2001parallel}
Cristopher Moore and Martin Nilsson.
\newblock Parallel quantum computation and quantum codes.
\newblock {\em SIAM Journal on Computing}, 31(3):799--815, 2001.

\bibitem{hoyer2005quantum}
Peter Hoyer and Robert Spalek.
\newblock Quantum fan-out is powerful.
\newblock {\em Theory of computing}, 1(5):83--101, 2005.

\bibitem{browne2010computational}
Dan Browne, Elham Kashefi, and Simon Perdrix.
\newblock Computational depth complexity of measurement-based quantum
  computation.
\newblock In {\em Conference on Quantum Computation, Communication, and
  Cryptography}, pages 35--46. Springer, 2010.

\bibitem{aaronson2005quantum}
Scott Aaronson.
\newblock Quantum computing, postselection, and probabilistic polynomial-time.
\newblock In {\em Proceedings of the Royal Society of London A: Mathematical,
  Physical and Engineering Sciences}, volume 461, pages 3473--3482. The Royal
  Society, 2005.

\bibitem{eldar2015local}
Lior Eldar and Aram~W Harrow.
\newblock Local hamiltonians whose ground states are hard to approximate.
\newblock {\em arXiv preprint arXiv:1510.02082}, 2015.

\bibitem{aaronson2004improved}
Scott Aaronson and Daniel Gottesman.
\newblock Improved simulation of stabilizer circuits.
\newblock {\em Physical Review A}, 70(5):052328, 2004.

\bibitem{amy2013meet}
Matthew Amy, Dmitri Maslov, Michele Mosca, and Martin Roetteler.
\newblock A meet-in-the-middle algorithm for fast synthesis of depth-optimal
  quantum circuits.
\newblock {\em IEEE Transactions on Computer-Aided Design of Integrated
  Circuits and Systems}, 32(6):818--830, 2013.

\bibitem{viola2014extractors}
Emanuele Viola.
\newblock Extractors for circuit sources.
\newblock {\em SIAM Journal on Computing}, 43(2):655--672, 2014.

\end{thebibliography}
\end{document}